\documentclass[english,11pt,a4paper,reqno]{amsart}
\usepackage{amstext,amsthm,amsfonts,amsmath,amssymb}
\usepackage{graphicx}
\usepackage{babel}
\usepackage{color} 
\usepackage{hyperref}
\usepackage{todonotes}
\usepackage{amsaddr}

\makeatletter

\newcommand{\D}{\mathcal{D}}

\renewcommand{\H}{\mathcal{H}}

\newcommand{\N}{\mathbb{N}}

\newcommand{\R}{\mathbb{R}}
\newcommand{\C}{\mathbb{C}}

\newcommand{\unit}{\mathbf{1}}
\newcommand{\ra}{\rightarrow}

\newcommand{\U}{\operatorname{U}}

\newcommand{\tr}{\operatorname{tr }}

\newcommand{\rmd}{\operatorname{d\!}}
\newcommand{\rmi}{\operatorname{i}}
\newcommand{\rme}{\operatorname{e}}

\renewcommand{\dim}{\operatorname{dim }}

\newtheorem{theorem}{Theorem}

\theoremstyle{definition}\newtheorem{definition}[theorem]{Definition}
\theoremstyle{definition}\newtheorem{remark}[theorem]{Remark}
\theoremstyle{definition}\newtheorem{example}[theorem]{Example}
\theoremstyle{definition}
\theoremstyle{definition}
\numberwithin{theorem}{section}

\makeatother
\title[Convergence of dynamical decoupling with finite amplitude controls]{Stability and convergence of dynamical decoupling with finite amplitude controls}
\author{Daniel Burgarth}
\address[DB]{Center for Engineered Quantum Systems, Macquarie University, Sydney 2109 NSW, Australia}
\author{Paolo Facchi}
\address[PF]{Dipartimento di Fisica, Universit\`{a} di Bari, I-70126 Bari, Italy, and\\ INFN, Sezione di Bari, I-70126 Bari, Italy}
\author{Robin Hillier}
\address[RH]{Department of Mathematics and Statistics, Lancaster University, Lancaster LA1 4YF, UK}
\email{r.hillier@lancaster.ac.uk}
\date{24 May 2022. Revised version 4 November 2022.}

\begin{document}

\maketitle
\thispagestyle{empty}

\begin{abstract}
Dynamical decoupling is a key method to mitigate errors in a quantum mechanical system, and we studied it in a series of papers dealing in particular with the problems arising from unbounded Hamiltonians. The standard bangbang model of dynamical decoupling, which we also used in those papers, requires decoupling operations with infinite amplitude, which is strictly speaking unrealistic from a physical point of view. In this paper we look at decoupling operations of finite amplitude, discuss under what assumptions dynamical decoupling works with such finite amplitude operations, and show how the bangbang description arises as a limit, hence justifying it as a reasonable approximation.
\end{abstract}

\section{Introduction}

The control of quantum noise and the suppression of decoherence in a quantum mechanical system have been the focus of quantum information theory for a long time, cf.~\cite{LB13} for an overview. On the one side there are error-correcting methods that allow to correct for errors by encoding the information in subspaces and correcting for leakage by measurements and gates. On the other side there are error-mitigation methods that aim to reduce errors before they happen by sequences of gates. Dynamical decoupling is a key tool in this regard. First presented in 1998 in~\cite{VL98,VKL99}, the underlying idea goes back to NMR~\cite{HW68,H76}. It has been studied and used by both theoretical and experimental physicists~\cite{LB13}.

The key idea of dynamical decoupling is to apply certain correction operations to a given quantum system interacting with a surrounding bath or environment in an unknown and unwanted way. These decoupling operations act on the system as unitary operators, and usually they are applied in a cyclic pattern and at constant frequency. The general hope is that as this frequency tends to infinity, decoherence tends to zero. A formal proof not taking into account domain issues of unbounded Hamiltonians can already be found in~\cite{VKL99}. Intuitively speaking, the idea is that for short times the system drifts slowly into some fixed direction, and these correction operations quickly rotate this fixed direction. These rotated drifts cancel with each other to first order and the net movement after some time is close to zero. Finding precise mathematical conditions and making rigorous statements as to when this does actually work is an open question. A first important step in that direction was made by us in~\cite{ABFH}, where we provided some necessary and sufficient conditions in order for dynamical decoupling to work, and we also studied a number of specific models in more detail, see also \cite{BFH}.

In most investigations, dynamical decoupling is understood to be the so-called \emph{bangbang dynamical decoupling}~\cite{VKL99}, where the correction operations are carried out in the form of instantaneous pulses with fixed frequency. This means that the correction operations need to have infinite amplitude, which makes the description unrealistic from a physical point of view. It therefore makes sense to look at correction operation which are implemented by finite-amplitude pulses of certain shapes that last for a nonzero amount of time~\cite{bang,
kicknfix}. We will call this \emph{continuous dynamical decoupling} in contrast to the usual bangbang dynamical decoupling.  In the case of bounded Hamiltonians this has been studied before~\cite{LB13} from various points of view. In this paper, we would like to look at potentially unbounded Hamiltonians on a joint system $\H=\H_s\otimes\H_e$, where $\H_s$ is a finite-dimensional system of interest and $\H_e$ is a potentially infinite-dimensional environment.

The first question we deal with here is what happens when the relative pulse length (i.e., the actual pulse length divided by the time difference between two consecutive operations) tends to zero, and hence the amplitude tends to infinity. One might suspect that the time evolution of the system should in some sense converge to the time evolution under bangbang dynamical decoupling. In Section~\ref{sec:approx}, we prove that this is indeed the case. This means that bangbang dynamical decoupling is not some abstract unphysical framework but instead a reasonable and mathematically simple approximation of the true situation, which among other things justifies previous research on bangbang dynamical decoupling.

The second question we ask is under what assumptions dynamical decoupling with finite amplitude pulses works directly, without having to consider the limit of the relative pulse length tending to zero. It turns out that the assumptions one has to make are stronger than in the bangbang case in order to overcome two obstacles. The first obstacle is analytic in nature and relates to domains of the bath Hamiltonians; hence it may arise in the case of infinite-dimensional environments only. The second obstacle is rather algebraic in nature and poses a problem in the case of finite-dimensional environments as well. 

In Section~\ref{sec:contDD}, we present a finite-dimensional example where the second obstacle is present and where continuous dynamical decoupling does not work although bangbang dynamical decoupling would work. However, in light of our first result, in the limit where the relative pulse length becomes shorter and shorter, the overall time evolution of this example becomes that of bangbang dynamical decoupling, which in turn does work. In Theorem~\ref{th:contDD} we study sufficient conditions in order to overcome the analytic obstacle, and we determine the resulting time evolution. In the case of bounded Hamiltonians, we can also provide concrete estimates on the speed of convergence depending on the norm of the Hamiltonian and a few other parameters, see Theorem~\ref{th:ratebounded}. 

Finally, in Section~\ref{sec:Euler} we deal with assumptions to make in order to overcome the second of the above obstacles, which is a constraint on the order of the decoupling operations. For this part we ask that the set of decoupling operations are arranged as a Cayley graph and a full decoupling cycle should be represented as an Euler cycle in this Cayley graph. In this case we speak of \emph{Euler dynamical decoupling}, which by construction is a special case of continuous dynamical decoupling. Euler dynamical decoupling was first introduced in~\cite{VK03} but without dealing with the sophisticated analytic problems of unbounded Hamiltonians. Here we provide sufficient conditions in order for Euler dynamical decoupling to work, and we close with an interesting example where this is indeed the case.

To summarise, we can say that continuous dynamical decoupling, or at least its more restrictive form of Euler dynamical decoupling, does indeed work under some fairly strong assumptions, even in the context of unbounded Hamiltonians. However, the description of continuous dynamical decoupling is technically more involved and moreover often only the weaker assumptions of bangbang dynamical decoupling are met, and we can say that in that case bangbang dynamical decoupling is a reasonable and useful approximation of reality.

\section{Preliminaries}\label{sec:prelim}

Let us consider a finite-dimensional quantum system with Hilbert space $\H_s$, say of dimension $d$, coupled to an environment with infinite-dimensional separable Hilbert space  $\H_e$, with the total space $\H :=\H_s\otimes\H_e$. The total Hamiltonian is given by a (possibly unbounded) densely defined selfadjoint operator $(H,\D(H))$.
We consider dynamical decoupling of this composite system, see~\cite[Def.2.1]{ABFH}. Moreover, we frequently write $u$ instead of $u\otimes\unit$, for $u\in\U(\H_s)$.

\begin{definition}\label{def:dd}
A \emph{decoupling set for $\H_s$} is a finite set of unitary operators $V\subset \U(\H_s)$ such that
\[
\frac{1}{|V|} \sum_{v\in V} v x v^* = \frac{1}{\dim \H_s} \tr(x) \unit_{\H_s},\quad \text{for all } x\in B(\H_s).
\]
Let $N$ be a multiple of the cardinality $|V|$. A \emph{decoupling cycle} of length $N$ is a cycle $(v_1, v_2,\ldots,v_N)$  through $V$ such that each element of $V$ appears the same number of times.
\end{definition}

\begin{remark}\label{rem:dd-def}
Sometimes we consider ``reduced" decoupling sets, namely sets $V$ such that the above equation holds only for specific given $x$. The idea behind this is that often one knows more about the Hamiltonian that has to be decoupled, hence does not need to work in complete generality; moreover, having a smaller number of decoupling operations may simplify the physical implementation and make the procedure more practicable.
\end{remark}

Running through the cycle over a time period $[0,t]$ results in the time evolution
\begin{eqnarray*}	
F(t) &:=&  \rme^{-\rmi \frac{t}{N} H} \gamma_{N} \cdots\rme^{-\rmi \frac{t}{N} H} \gamma_2 \rme^{-\rmi \frac{t}{N} H} \gamma_1
\\
&=& v_N^* \bigl(\rme^{-\rmi \frac{t}{N} v_N H v_N^*}  \cdots\rme^{-\rmi \frac{t}{N} v_2 H v_2^*}  \rme^{-\rmi \frac{t}{N} v_1 H v_1^*} \bigr) v_N,
\end{eqnarray*}
where $\gamma_{k}:= v_{k}^* v_{k-1}$, for $k=1,\dots,N$, with $v_0 := v_N$. Write $\Gamma\subset \U(\H_s)$ for the set of all those $\gamma_k$.
Notice that there is a gauge freedom in the choice of the cycle. Indeed,  $(v_1, v_2, \dots, v_n)$ is equivalent to the cycle $(\tilde{v}_1, \tilde{v}_2, \dots, \tilde{v}_n)$ with $\tilde{v}_k = w^* v_k$ for an arbitrary unitary $w\in V$: they generate the same control sequence, $\gamma_{k}= v_{k}^* v_{k-1} = \tilde{v}_{k}^* \tilde{v}_{k-1}$ and if one is decoupling so is the other. Therefore, without loss of generality, in the following we will always fix the gauge $w=v_N$, i.e.\ we will assume that the last element of the cycle is the identity, $\tilde{v}_N=\unit$ (and drop the tilde).

 We repeat the cycle several times and hence write $v_{k+j N} := v_k$, for every $k\in\{1,\ldots,N\}$ and $j\in\N$. If the  cycle is applied $m$ times over a time interval $[0,t]$ with (standard) instantaneous decoupling operations, with the operation $\gamma_k=v_{k}^* v_{k-1}$ applied at time $\frac{(k-1)t}{mN}$, we have the total time evolution (in the canonical gauge)
\begin{equation}\label{eq:Ftm-bb}
\begin{aligned}
F\Big(\frac{t}{m}\Big)^m =& \rme^{-\rmi \frac{t}{m N} H} \gamma_{m N} \cdots\rme^{-\rmi \frac{t}{m N} H} \gamma_2 \rme^{-\rmi \frac{t}{m N} H} \gamma_1 
	\\
	=&	
	\bigl(\rme^{-\rmi \frac{t}{mN} v_{mN} H v_{mN}^*}  \cdots\rme^{-\rmi \frac{t}{mN} v_2 H v_2^*}  \rme^{-\rmi \frac{t}{mN} v_1 H v_1^*} \bigr). 
\end{aligned}
\end{equation}
We say that \emph{bangbang dynamical decoupling works} for this given decoupling set $V$ and decoupling cycle $(v_1,\ldots,v_N)$ if there is a selfadjoint $B$ on $\H_e$ such that
\[
F\Big(\frac{t}{m}\Big)^m \ra \unit\otimes \rme^{-\rmi t B},\quad m\ra\infty,
\]
strongly and uniformly for $t$ in compact intervals. For a slightly weaker notion of decoupling, see~\cite{BFFH}.

Since these instantaneous pulses require infinite amplitude, they are strictly speaking unphysical and one considers a finite amplitude analogue, namely consider a smooth path $\gamma_k(s)$ in $U(\H_s)$ such that 
\begin{equation}\label{eq:gammak}
\gamma_k(s) = \left\{
\begin{array}{l@{\: :\:}l}
\unit & s\leq 0 \\
v_{k}^*v_{k-1} & s\geq 1 .
\end{array}\right.
\end{equation}
Moreover, let $A_k = A_k^*\in B(\H_s)$ be such that $\gamma_k = v_{k}^*v_{k-1} = \rme^{-\rmi A_k}$, and consider the path 
\begin{equation}\label{eq:Ak}
A_k(s):= \rmi \gamma_k'(s) \gamma_k(s)^* \in B(\H_s),
\end{equation}
with $A_k(s) = A_k(s)^*$. Physically, this describes the shape and direction of the smoothed out dynamical decoupling pulse. Often, e.g.~in Fig.~\ref{fig:figure}, one might be interested in $A_k(s) = \varphi_k(s) A_k$ with $\varphi_k$ an
integrable function with support in $[0,1]$ and $\int_\R \varphi_k(s) \rmd s =1$ but we try to keep things as general as possible. Now for every fixed  $k$ and $\tau>0$, consider the time-dependent Schr\"odinger equation
\begin{equation}\label{eq:timeSchroedinger}
\frac{\rmd}{\rmd s} u_k(\tau;s) = -\rmi \Big(\frac{1}{\tau} A_k(s/\tau) + H\Big) u_k(\tau;s), \quad 
u_k(\tau;0)=\unit, \qquad s\in[0,\tau].
\end{equation}
For time-dependent unbounded Hamiltonians, it is in general not obvious that solutions to the Schr\"odinger equation exist. However in this simple case with $H$ time-independent and $A_k(s)$ bounded, we know from~\cite[Thm.5.3.1]{Pazy} or~\cite[Thm.1.2]{Kato85} that a strongly continuous solution,
which preserves $\D(H)$, denoted by the time-ordered exponential
\[
u_k(\tau;s) = \mathop{\overleftarrow{\exp}}\Bigl(-\rmi  \int_{0}^{s} \Bigl[\frac{1}{\tau} A_k(r/\tau) +  H\Bigr]\rmd r \Bigr), \qquad s\in[0,\tau],
\]
exists. In analogy to the previous case, running once through the complete continuous decoupling cycle over a time period $[0,t]$ results in the time evolution
\begin{equation}\label{eq:Ft}
F(t):= u_N(t/N;t/N)\cdots  u_2(t/N;t/N)\, u_1(t/N;t/N).
\end{equation}
In analogy to the previous definition, we say that \emph{continuous dynamical decoupling works} for this system with given decoupling set $V$, decoupling cycle $(v_1,\ldots,v_N)$, and pulse shapes $\gamma_k(s)$  if there is a selfadjoint $B$ on $\H_e$ such that
\[
F\Big(\frac{t}{m}\Big)^m \ra \unit\otimes \rme^{-\rmi t B},\quad m\ra\infty,
\]
strongly and uniformly for $t$ in compact intervals.

\section{Approximation of instantaneous singular decoupling pulses}\label{sec:approx}

This section deals with the question whether bangbang dynamical decoupling is a good approximative description of the actual physical continuous dynamical decoupling. We will not make any assumption as to whether continuous dynamical decoupling works but we will assume that bangbang dynamical decoupling works, for a given system and dynamical decoupling cycle.

We consider the time evolution unitary
\begin{equation}\label{eq:BB-evol}
\Big(  \rme^{-\rmi \frac{t}{mN} H} \gamma_{N} \cdots   \rme^{-\rmi \frac{t}{mN} H} \gamma_1\Big)^m
= \Big(\rme^{-\rmi\frac{t}{mN}H} \rme^{-\rmi A_N} \cdots \rme^{-\rmi\frac{t}{mN}H} \rme^{-\rmi A_1} \Big)^m
\end{equation}
under bangbang (instantaneous) dynamical decoupling, where we repeat a given decoupling cycle $(v_1,\ldots, v_N)$ $m$-times over a time period $[0,t]$, realised through pulses $\gamma_k=v_{k}^* v_{k-1} = \rme^{-\rmi A_k}$, with $k=1,\ldots,N$, as introduced above.
\begin{figure}
\includegraphics[width=0.78\textwidth]{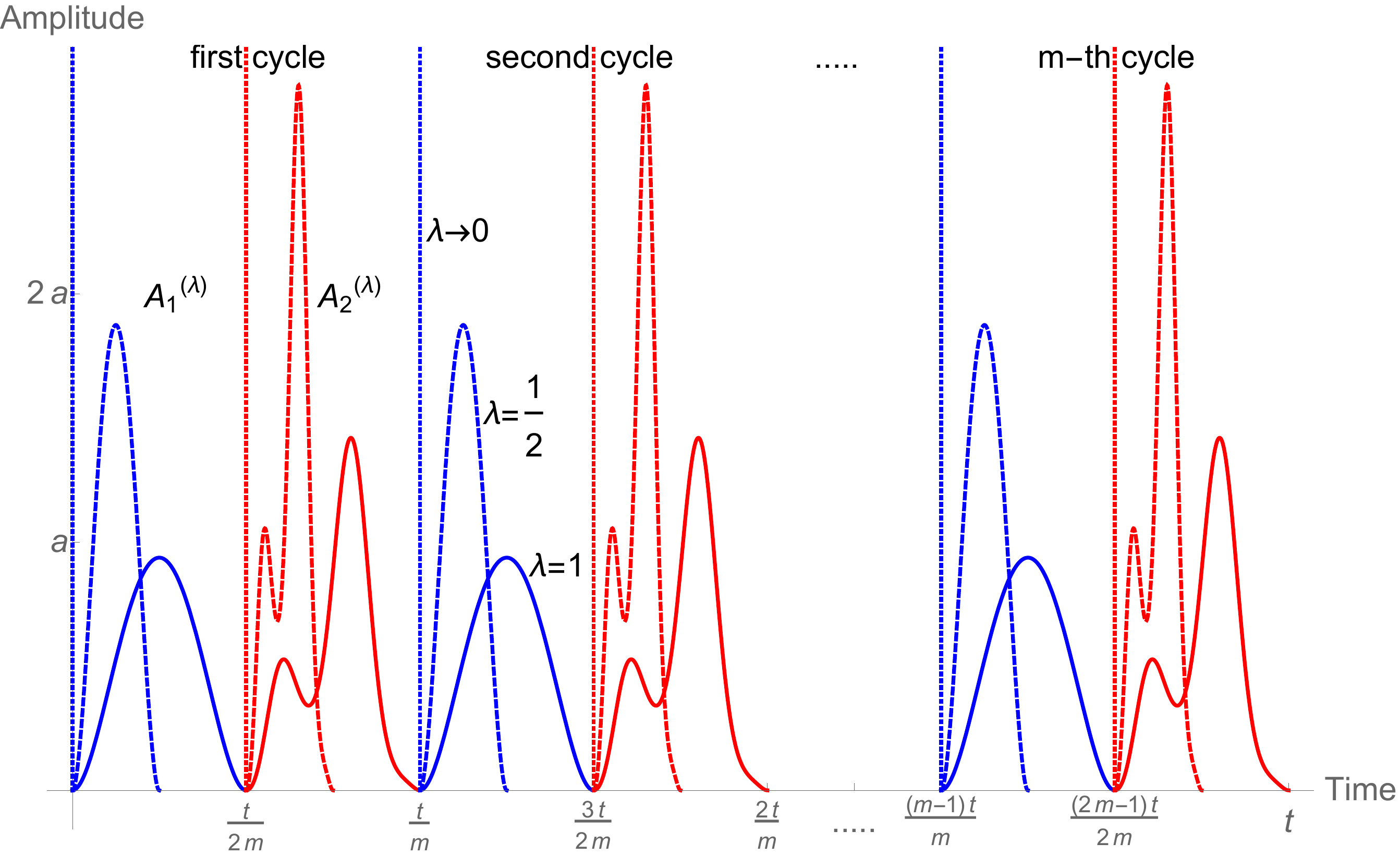}
\caption{An example of decoupling pulses for a cycle of length $N=2$, one element of the cycle being denoted by blue colour, the other one by red colour. For fixed overall time $t$, we repeat the cycle $m$ times. For each of the two elements of the decoupling cycle we choose different pulse shapes $A_k^{\lambda}(s)$, with $k=1,2$. We show the dependence on the scaling parameter $\lambda$, namely in addition to $\lambda=1$ (full line), we include the case $\lambda=\frac12$ (dashed line) and the limit $\lambda\to 0$ (singular dotted line). Each pulse has support on an interval of length $\frac{\lambda t}{mN}$ and the area under each pulse is independent of $\lambda$. As $\lambda\rightarrow 0$, the pulses become $\delta$ pulses with infinite amplitude, which corresponds to bangbang dynamical decoupling. Theorem~\ref{th:DD-convergence} considers this convergence, while Theorems \ref{th:contDD} and \ref{th:Euler} work with the case $\lambda>0$ corresponding to continuous dynamical decoupling.
\label{fig:figure}}
\end{figure}

Given the smooth function $A_k(s)$ from~\eqref{eq:Ak}, for every $\lambda\in(0,1]$, let 
\begin{equation}\label{eq:philambda}
A_k^{(\lambda)}(s) := \frac{1}{\lambda} A_k\Big(\frac{s}{\lambda}\Big)
\end{equation}
and $\gamma_{k,\lambda}(s):=\gamma_k(s/\lambda)$.
Then we consider the time-dependent Hamiltonian
\begin{equation}\label{eq:Adeltaalpha}
H_\lambda(s) = H + \sum_{k=1}^{N}  A_k^{\left(\frac{\lambda t}{N}\right)}\Bigl(s- (k-1) \frac{t}{N}\Bigr)  , \quad s\in[0,t],
\end{equation}
with constant domain $\D(H_\lambda(s))=\D(H)$. The $k$-th pulse  $A_{k}^{(\lambda t/N)}(s-(k-1) t/N)$ 
has finite width $\lambda\, t/N$ and is supported in $[(k-1) t/N, (k-1+\lambda) t/N]$, as shown in Fig.~\ref{fig:figure}.
We can see that the conditions in~\cite[Thm.5.3.1]{Pazy}, or alternatively the slightly simplified conditions in~\cite[Thm.1.2]{Kato85}, are fulfilled, individually for each $k$, with $Y=\D(H)$ and $\|\cdot\|_Y$ the graph norm of $H$, hence the time-dependent Schr\"odinger equation for $H_\lambda(s)$ has got a strongly continuous solution, and hence the corresponding analogue of~\eqref{eq:Ft} is 
\begin{equation}
	F_\lambda(t):= u_N(t/N,\lambda;t/N)\cdots  u_2(t/N,\lambda;t/N)\, u_1(t/N,\lambda;t/N),
\end{equation}
where
\[
	u_k(\tau,\lambda;s)= \mathop{\overleftarrow{\exp}}\Bigl(-\rmi  \int_{0}^{s} \Bigl[A_{k}^{(\lambda\tau)}(r) +  H\Bigr]\rmd r \Bigr), \qquad s\in[0,\tau].
\]
The time evolution after $m$ decoupling cycles is then $F_\lambda(t/m)^m$.

We have:

\begin{theorem}\label{th:DD-convergence}
Suppose that bangbang dynamical decoupling works for a given quantum system with decoupling cycle. Then dynamical decoupling with finite-amplitude pulses $A_{k}^{(\lambda)}$ works arbitrarily well by choosing $\lambda>0$ sufficiently small. More precisely, there is a selfadjoint $B$ on $\H_e$ such that
\[
\lim_{m\to\infty}\lim_{\lambda\to 0} F_\lambda(t/m)^m \xi = \unit\otimes \rme^{-\rmi t B} \xi,
\]
for all $\xi\in\H$ and uniformly for $t$ in compact intervals.
\end{theorem}

Roughly speaking, the theorem tells us that, under reasonable assumptions, the bangbang model is an adequate mathematical approximation of a physical setup with finite amplitude decoupling pulses.

\begin{proof}
First we notice that, for every $\xi\in\H$,
\begin{equation}\label{eq:U-triangle1}
\begin{aligned}
\|F_\lambda(t/m)^m \xi - \unit\otimes \rme^{-\rmi t B} \xi \| \le& \|F_\lambda(t/m)^m \xi - F(t/m)^m\xi \|  + \|F(t/m)^m\xi  - \unit\otimes \rme^{-\rmi t B} \xi\|,
\end{aligned}
\end{equation}
where $F(t/m)^m$ has been defined in~\eqref{eq:Ftm-bb}. We choose $B$ and $m$ such that the second term on the RHS becomes sufficiently small, which is possible if the standard bangbang dynamical decoupling scheme works. Moreover, for this $m$, we will show below that we can then choose $\lambda>0$ small enough such that the first term on the RHS becomes sufficiently small, which proves the theorem.

Now rather than looking at the whole interval $[0,t]$ with $mN$ decoupling operations, we would like to look at one of the subintervals $[(k-1)\tau,k \tau]$, of width $\tau= t/(mN)$. The overall time evolution is a product of the time evolutions over these $mN$ subintervals, so if we can show strong convergence on these subintervals,
\[
u_k(\tau,\lambda;\tau)\,\xi \to \rme^{-\rmi \tau H} \gamma_k \, \xi, \qquad \lambda\downarrow 0,
\]
for all $\xi\in\H$, then by the strong continuity of multiplication it follows for the total time evolution. We get
\begin{eqnarray*}
u_k(\tau,\lambda;\tau) &=& \mathop{\overleftarrow{\exp}}\Bigl(-\rmi  \int_{0}^{\tau} \Bigl[A_{k}^{(\lambda\tau)}(s) +  H\Bigr]\rmd s \Bigr)
= \mathop{\overleftarrow{\exp}}\Bigl(-\rmi  \int_{0}^{1} \Bigl[A_{k}^{(\lambda)}(s) +  \tau H\Bigr]\rmd s \Bigr)
\\
&=& \rme^{-\rmi \tau H} \mathop{\overleftarrow{\exp}}\Bigl(-\rmi  \int_{0}^{1} \rme^{\rmi s \lambda\tau H} A_{k}(s) \rme^{-\rmi s \lambda\tau H}\rmd s \Bigr).
\end{eqnarray*}
Thus all we need to show is that
\[
\mathop{\overleftarrow{\exp}}\Bigl(-\rmi  \int_{0}^{1} \rme^{\rmi s \lambda\tau H} A_{k}(s) \rme^{-\rmi s \lambda\tau H}\rmd s \Bigr) \to \mathop{\overleftarrow{\exp}}\Bigl(-\rmi  \int_{0}^{1}  A_{k}(s) \rmd s \Bigr) = \gamma_k, 
\]
strongly, as $\lambda\to 0$.

Recall, e.g.~from~\cite[Eqn.(5.3)]{Kato85}, that if $U_j(t)= \mathop{\overleftarrow{\exp}}(-\rmi \int_{0}^{t} H_j(s) \rmd s)$, are the unitary propagators generated by two bounded selfadjoint time-dependent operators $H_j(t)$, for $j=1,2$, then 
$$
U_1(t)-U_2(t) = -\rmi \int_0^t U_1(t)U_1(s)^* \big(H_1(s)-H_2(s)\big) U_2(s)\rmd s,
$$
whence
$$
\bigl\|\big(U_1(t)-U_2(t)\big)\xi\bigr\| \leq \int_0^t \bigl\|  \big(H_1(s)-H_2(s)\big) U_2(s) \xi \bigr\|\rmd s.
$$
Applying this with
\[
H_1(s)=\rme^{\rmi s \lambda\tau H} A_{k}(s) \rme^{-\rmi s \lambda\tau H},\quad
H_2(s)=A_{k}(s),
\]
and using our above computations, we obtain
\begin{eqnarray*}
\bigl\|\big(u_k(\tau,\lambda;\tau)-\rme^{-\rmi \tau H} \gamma_k\big) \xi \bigr\|	
&=& \bigl\|\big(\rme^{\rmi \tau H} u_k(\tau,\lambda;\tau)- \gamma_k\big) \xi \bigr\| \\
&\leq& \int_{0}^{1} \bigl\| \big( \rme^{\rmi s \lambda\tau H} A_{k}(s) \rme^{-\rmi s \lambda\tau H}
- A_{k}(s)  \big) \gamma_k(s) \xi  \bigr\| \rmd s
\\
&=&  \int_{0}^{1} \bigl\| \bigl[  A_{k}(s) , \, \rme^{-\rmi s \lambda\tau H} \bigr]\, \gamma_k(s)  \xi  \bigr\| \rmd s \to 0,
\end{eqnarray*}
as $\lambda\downarrow 0$, by dominated convergence, using the strong continuity of $\lambda\mapsto\rme^{-\rmi s\lambda\tau H}$ and hence pointwise convergence of the integrand  together with its boundedness by $2\|A_k\|_\infty \|\xi\|$.

So far this was for one specific choice of $k$. But $F_\lambda(t/m)$ is actually an $mN$-fold product, where strong convergence holds for each of the $mN$ factors. Due to the strong continuity of multiplication of uniformly bounded factors,
 we see that the first term on the RHS of~\eqref{eq:U-triangle1} can be made arbitrarily small as $\lambda\downarrow 0$.

\end{proof}

\begin{remark}\label{rem:DD-convergence}
If $\D\subset\D(H)$ is a subspace which is invariant under all $\gamma_k(s)$ and $A_k(s)$ and such that 
\[
s\in[0,1]\mapsto \big\|HA_k(s)\gamma_k(s)\xi\big\|, \quad \big\| H\gamma_k(s)\xi \big\|
\]
are integrable for every $\xi\in\D$, the above proof also gives us a simple estimate on the speed of convergence in Theorem~\ref{th:DD-convergence} relative to a given vector. Namely, for every $\xi\in\D$, we have:
\begin{align*}
\bigl\|\big(u_k(\tau,\lambda;\tau)-\rme^{-\rmi \tau H} \gamma_k\big) \xi \bigr\|	
\leq&  \int_{0}^{1} \bigl\| \bigl[  A_{k}(s) , \, \rme^{-\rmi s \lambda\tau H} \bigr]\, \gamma_k(s)  \xi  \bigr\| \rmd s \to 0\\
=& \int_{0}^{1} \| A_{k}(s)\big\|\, \big\|\big(\rme^{-\rmi s \lambda\tau H} -\unit\big) \gamma_k(s)  \xi  \big\| \rmd s \\
& \quad + \int_{0}^{1} \bigl\| \big( \rme^{-\rmi s \lambda\tau H} -\unit \big) A_k(s)\gamma_k(s)\xi \bigr\| \rmd s\\
\leq& \lambda\tau \int_0^1  \Big( \big\| A_{k}(s)\big\|\,\big\| H\gamma_k(s)\xi \big\| + \big\|HA_k(s)\gamma_k(s)\xi\big\|\Big) \rmd s,
\end{align*}
which is linear in $\lambda\tau$. Taking the sum over $k$ removes $\tau$ and turns this into an upper bound proportional to $\lambda$, so the convergence to bangbang dynamical decoupling is of order $O(\lambda)$.
\end{remark}

\section{Continous decoupling}\label{sec:contDD}

One may ask the question whether continuous dynamical decoupling works automatically provided bangbang dynamical decoupling works. The following example shows that this is not the case.

\begin{example}\label{Counterexample}
Let us consider the simplest case of a qubit, with trivial environment $\H_e=\C$, so our Hilbert space is $\mathcal{H}=\mathcal{H}_s=\mathbb{C}^2.$ This means that the total Hamiltonian $H$ is bounded. Let us denote the Pauli matrices as $\{X,Y,Z\}.$ The standard decoupling set $V$ for which bangbang dynamical decoupling works is given by the finite group generated by the Pauli matrices. Due to cancellation of phases it suffices to focus on the decoupling cycle $(X,Y,Z,\unit)$ of length $N=4$. The bangbang evolution is given by 
\[
 \Big( \unit\rme^{-\rmi \frac{t}{4m}H} \unit Z \rme^{-\rmi \frac{t}{4m}H}ZY\rme^{-\rmi \frac{t}{4m}H}YX\rme^{-\rmi \frac{t}{4m}H}X \Big)^m,
\]
and this tends to $\unit$ as $m\to\infty$.
We can use the properties of Pauli matrices to simplify this expression as
\[
(-1)^m\Big( \rme^{-\rmi \frac{t}{4m}H}Z \rme^{-\rmi \frac{t}{4m}H} X \Big)^{2m}
=(-1)^m\Big( \rme^{-\rmi \frac{t}{4m}H}\rme^{-\rmi \frac{\pi}{2}Z} \rme^{-\rmi \frac{t}{4m}H}  \rme^{-\rmi \frac{\pi}{2}X} \Big)^{2m}
\]
Let us imagine instead that the experimentalist performs the sequence in a continuous fashion with rectangular pulses of full length $t/(4m)$, which means $\lambda=1$ in the notation of the previous section. The overall time evolution then becomes
\[
F_1\Big( \frac{t}{m}\Big) ^m = (-1)^m\Big( \rme^{-\rmi \frac{\pi}{2}Z -\rmi \frac{t}{4m}H}  \rme^{-\rmi \frac{\pi}{2}X -\rmi \frac{t}{4m}H}\Big)^{2m} .
\] 
For large $m$, we can split the exponentials as
\[
\rme^{-\rmi \frac{\pi}{2}\sigma -\rmi \frac{t}{4m}H} = \rme^{-\rmi \frac{\pi}{2} \sigma} \rme^{-\rmi \frac{t}{4m} H_\sigma} + \mathcal{O}\left(\frac{1}{m}\right),
\]
where $H_\sigma$ is the Zeno Hamiltonian~\cite{unity1} of $H$ with respect to the strong part $\sigma=X,Z$. The hope is that this recovers dynamical decoupling even with the pulses spread out. However, the error term might accumulate to something of $\mathcal{O}(1)$ under exponentiation, and we will show now that this is indeed sometimes the case.

To this end, consider the choice $H=X$, so that
\[
F_1\Big( \frac{t}{m}\Big) ^m = (-1)^m\Big( \rme^{-\rmi \frac{\pi}{2}Z -\rmi \frac{t}{4m}X}  \rme^{-\rmi \frac{\pi}{2}X -\rmi \frac{t}{4m}X}\Big)^{2m}.
\] 
In the large $m$ limit only terms up to order $1/m$ contribute, so we neglect all higher order terms. 

For the first term on the right-hand side, we have
\[
\rme^{-\rmi \frac{\pi}{2}Z -\rmi \frac{t}{4m}X}
=-\rmi Z- \rmi \frac{t}{4m}\int_0^1 \rme^{-\rmi (1-s) \frac{\pi}{2} Z} X  \rme^{-\rmi s \frac{\pi}{2} Z}\rmd s +\mathcal{O}\big(\frac{1}{m^2}\big)
= -\rmi Z -\rmi \frac{t}{2\pi m} X +\mathcal{O}\big(\frac{1}{m^2}\big)
\] 
while for the second term we get
\[
\rme^{-\rmi \frac{\pi}{2}X -\rmi \frac{t}{4m}X} = 
(-\rmi X) \Big( \unit -\rmi \frac{t}{4m} X + \mathcal{O}\big(\frac{1}{m^2}\big)\Big).
\]
Thus
\begin{align*}
\lim_{m\rightarrow\infty} F_1\Big( \frac{t}{m}\Big) ^m
=& \lim_{m\rightarrow \infty}(-1)^m \Big( \big( -\rmi Z - \rmi \frac{t}{2\pi m} X \big)(-\rmi X) \big( \unit - \rmi \frac{t}{4m} X\big)\Big)^{2m}\\
=& \lim_{m\to\infty} (-1)^m \Big( -\rmi Y- \frac{t}{2\pi m} \unit + \rmi \frac{t}{4m}Z \Big)^{2m}
\end{align*}
Finally, squaring out the right-hand side and using anticommutativity of Pauli matrices we obtain
\[
\lim_{m\to\infty} F_1\Big( \frac{t}{m}\Big) ^m
=\lim_{m\to\infty}(-1)^m\left (-\unit+\rmi\frac{t}{\pi m}Y\right)^{m}=\rme^{-\rmi \frac{t}{\pi}Y},
\] 
which is different from a multiple of $\unit$, in general. So indeed, continuous dynamical decoupling does not work for the standard decoupling set.
\end{example}

The next theorem provides some sufficient criteria as to whether $\lim_{m\to\infty} F_\lambda\big( \frac{t}{m}\big)^m$ exists in the case of continuous pulses and what it is equal to in general:

\begin{theorem}\label{th:contDD}
Consider a (not necessarily decoupling) cycle $(v_1, v_2, \dots, \unit)$ through $U(\H_s)$. Assume that
\begin{itemize}
\item[(i)] there is a dense subspace $\D\subset\D(H)$ which is invariant under $u_k(t/N,\lambda;t/N)$ and $\gamma_k(s)$, for all $k,s,t$, and such that
$
s\in[0,1]\mapsto \|H\gamma_k(s)\xi\|
$
is integrable, for all $\xi\in\D$;
\item[(ii)] the operator
\[
\frac{1-\lambda}{N} \sum_{k=1}^N v_k H v_k^* +  \frac{\lambda}{N} \sum_{k=1}^N v_{k-1} \int_{0}^1 \gamma_k(s)^* H\gamma_k(s) \rmd s \; v_{k-1}^*
\]
is essentially selfadjoint on $\D$, and we denote its closure by $H_\lambda$.
\end{itemize}
Then
\[
F_\lambda \Big(\frac{t}{m}\Big) ^m\xi \to 
\rme^{-\rmi t H_\lambda} 
\xi, \quad m\to\infty,
\]
for all $\xi\in\H$ and uniformly for all $t$ in compact intervals.
\end{theorem}

\begin{remark}\label{rem:dd-TK}
Under the assumptions of Theorem \ref{th:contDD} for all $\lambda\in[0,1]$, it follows from the Trotter-Kato theorem \cite[Th.3.4.8]{EN} that
\[
\rme^{-\rmi t H_\lambda}\xi \to \rme^{-\rmi t H_0}\xi, \quad \lambda\to 0,
\]
for all $\xi\in\H$ and uniformly for all $t$ in compact intervals. Comparing with Theorem \ref{th:DD-convergence} we thus have
\[
\lim_{m\to\infty} \lim_{\lambda\to 0} F_\lambda \Big(\frac{t}{m}\Big) ^m\xi = \lim_{\lambda\to 0}\lim_{m\to\infty} F_\lambda \Big(\frac{t}{m}\Big) ^m\xi,
\]
for all $\xi\in\H$ and all $t\in\R$, i.e., we can swap the limits of letting the number of repetitions of the decoupling cycle go to infinity and that of the relative pulse length going to $0$.
\end{remark}

\begin{proof}
Using (i), we see that $F(t)$ as defined in~\eqref{eq:Ft} leaves $\D$ invariant and hence, by setting $u_k:=u_k(t/N,\lambda;t/N)$, we obtain:
\begin{align*}
\frac{F_\lambda(t)-\unit}{t} \xi =& 
\frac{\prod_{k=1}^N u_k(t/N,\lambda;t/N)-\unit}{t}\xi\\
=& \frac{\prod_{j=1}^N u_j  - \prod_{l=1}^N \gamma_{l}}{t}\xi 
\\
=& \frac{\prod_{j=1}^N  u_j -(\prod_{j=2}^N u_j) \gamma_1}{t}\xi + \frac{(\prod_{j=2}^N u_j) \gamma_1  -(\prod_{j=3}^N  u_j) \gamma_2\gamma_1}{t}\xi 
\\
&+ \ldots +
\frac{u_N (\prod_{l=1}^{N-1} \gamma_{l}) - \prod_{l=1}^N \gamma_{l}}{t}\xi\\
=& \sum_{k=1}^N \Bigl(\prod_{j=k+1}^N u_j\Bigr) \frac{u_k-\gamma_{k}}{t} \Bigl(\prod_{l=1}^{k-1} \gamma_{l}\Bigr) \xi \\
=& \sum_{k=1}^N \Bigl(\prod_{j=k+1}^N u_j(t/N,\lambda;t/N)\Bigr) \frac{u_k(t/N,\lambda;t/N)-\gamma_{k}}{t} \, v^*_{k-1} 
\xi \\
\end{align*}
is well-defined on $\D$. 

We have
\[
u_k(t/N,\lambda;t/N) = \rme^{-\rmi  \frac{t}{N}(1-\lambda)H}\gamma_k\mathop{\overleftarrow{\exp}}\Big(-\rmi  \frac{\lambda t}{N} \int_{0}^{1} \gamma_k(s)^* H \gamma_k(s) \rmd s\Big),
\]
and the first term on the right-hand side is clearly differentiable with respect to $t$. In order to study the second one, let us make the following consideration: if $U_\alpha(t,s)$ solves the differential equation
\[
\frac{\rmd}{\rmd s} U_\alpha(t,s) = \rmi U_\alpha(t,s) \alpha H(s)
\]
with $\alpha\in\R$ then it also solves the following integral equation
\[
U_\alpha(t,s)-\unit = -\rmi\alpha \int_s^t U_\alpha(t,r)H(r) \rmd r.
\]
Thus
\begin{align*}
\Big\| \frac{U_\alpha(1,0)-\unit}{\alpha} \xi + & \rmi \int_0^1 H(r)\xi \rmd r \Big\|
= \Big\| -\rmi \int_0^1 \big(U_\alpha(1,r)-\unit \big) H(r)\xi \rmd r \Big\| \\
\leq& \int_0^1  \Big\| \big(U_\alpha(1,r)-\unit\big)  H(r)\xi\Big\| \rmd r
\end{align*}
for suitable $\xi$. Now the integrand is dominated by $2\|H(r)\xi\|$, and we require $\xi$ to be in all $\D(H(r))$ and such that this is integrable on $[0,1]$; furthermore, the integrand converges to $0$ as $\alpha\to 0$, pointwise for every $r\in[0,1]$. Thus by Lebesgue's dominated convergence theorem we get that the right-hand side converges to $0$ as $\alpha\to 0$, in other words
\[
\frac{U_\alpha(1,0)-\unit}{\alpha} \xi \to - \rmi \int_0^1 H(r)\xi \rmd r.
\]
We apply this to our setting with $\alpha=t$ and
\[
H(r)= \frac{\lambda}{N}\gamma_k(r)^*H\gamma_k(r)
\]
to obtain the differentiability of $U_t(1,0)=\mathop{\overleftarrow{\exp}}\Big(-\rmi  \frac{\lambda t}{N} \int_{0}^{1} \gamma_k(s)^* H \gamma_k(s) \rmd s\Big)$ w.r.t.~$t$, and
hence by the product rule
\[
\frac{u_k(t/N,\lambda;t/N)-\gamma_{k}}{t}\xi \to
 \Big(-\rmi  \frac{\lambda}{N} \gamma_{k} \int_{0}^{1} \gamma_{k}(s)^* H \gamma_{k}(s) \rmd s-\rmi  \frac{1-\lambda}{N} H \gamma_{k}\Big)\xi,
\]
as $t\to 0$, for all $\xi\in\D$ as in assumption (i). Hence, we obtain
\[
\lim_{t\ra 0} \frac{F_\lambda(t)-\unit}{t} \xi = \sum_{k=1}^N 
v_k  \Big( -\rmi  \frac{\lambda}{N} \gamma_{k} \int_{0}^{1} \gamma_{k}(s)^* H \gamma_{k}(s) \rmd s-\rmi  \frac{1-\lambda}{N} H \gamma_{k}\Big) v_{k-1}^* 
\xi,
\]
for all $\xi\in\D$.

By assumption (ii),
\begin{equation}\label{eq:HlambdaDD}
\begin{aligned}
\sum_{k=1}^N v_k  \Big( \frac{\lambda}{N} & \gamma_{k} \int_{0}^{1} \gamma_{k}(s)^* H \gamma_{k}(s) \rmd s + \frac{1-\lambda}{N} H \gamma_{k}\Big) v_{k-1}^*\\
=& \frac{\lambda}{N}\sum_{k=1}^N   v_{k-1} \Big( \int_{0}^{1} \gamma_k(s)^* H \gamma_k(s) \rmd s\Big) v_{k-1}^* +  \frac{1-\lambda}{N}\sum_{k=1}^N v_k  H v_k^*
\end{aligned}
\end{equation}
is essentially selfadjoint on $\D$ with closure $H_\lambda$. We can apply Chernoff's theorem~\cite[Th.11.1.2]{Che} to obtain that
\[
F_\lambda\Big(\frac{t}{m}\Big)^m \xi \ra 
\rme^{-\rmi t H_\lambda} 
\xi, \quad m\ra\infty,
\]
uniformly for all $t$ in compact intervals and $\xi\in\H$.
\end{proof}
\begin{remark}
In the special case where $\lambda=0$ and $(v_1,\ldots,v_N)$ is a decoupling cycle, we are in the situation of bangbang dynamical decoupling and we see that the first sum on the right-hand side in~\eqref{eq:HlambdaDD} vanishes and we obtain $H_0=\unit \otimes B$ for some selfadjoint operator $B$ on $\H_e$, hence dynamical decoupling works. For general $\lambda$ this is not the case, as Example~\ref{Counterexample} showed, where $H_1=\frac{1}{\pi}Y\not=\unit\otimes B$.
\end{remark}

In the special case of a bounded Hamiltonian, we can provide estimates on the speed of convergence in Theorem~\ref{th:contDD}:

\begin{theorem}\label{th:ratebounded}
If $H$ is bounded we have
\begin{equation}\label{eq:ratebounded1}
\Big\| F_\lambda\Big(\frac{t}{m}\Big)^m - \rme^{-\rmi t H_\lambda}\Big\| \leq \frac{2t}{m} \| H\| (1+2t \|H\|)
\end{equation}
and
\begin{equation}\label{eq:ratebounded2}
\Big\| F_\lambda\Big(\frac{t}{m}\Big)^m - \rme^{-\rmi t H_0}\Big\| \leq \frac{2t}{m} \| H\| (1+2t \|H\|) + \lambda t \|H_1-H_0\|,
\end{equation}
with $H_\lambda$ as in Theorem~\ref{th:contDD}.
\end{theorem}

\begin{proof}
We start with some preparation: let us write
\[
H_\lambda^C(s) = \sum_{k=1}^{mN}  A_{k}^{\left(\frac{\lambda t}{mN}\right)}\Bigl(s-(k-1) \frac{t}{mN}\Bigr)  , \quad s\in[0,t]
\]
for the part of $H$ in~\eqref{eq:Adeltaalpha} that describes the control operations, and
\[
U_\lambda^C(s) = \mathop{\overleftarrow{\exp}} \Big(-\rmi \int_0^s H_\lambda^C(r)\rmd r \Big), \quad s\in[0,t]
\]
for the corresponding unitary time evolution. Then
\begin{align*}
\int_{(k-1)t/mN}^{kt/mN} U_{\lambda}^C(s)^* HU_{\lambda}^C(s) \rmd s 
=& \frac{t}{mN} \int_{0}^1 v_{k-1}\gamma_{k,\lambda}(s)^* H \gamma_{k,\lambda}(s)v_{k-1}^* \rmd s 
\\
=& \frac{t}{mN} \int_0^{\lambda} v_{k-1}\gamma_{k,\lambda}(s)^* H \gamma_{k,\lambda}(s)v_{k-1}^* \rmd s  
+ \frac{t}{mN} \int_{\lambda}^{1} v_{k-1} \gamma_k^* H \gamma_k v_{k-1}^* \rmd s
\\
=& \lambda\frac{t}{mN} \int_{0}^1 v_{k-1}\gamma_{k}(s)^* H \gamma_{k}(s)v_{k-1}^* \rmd s
+ (1-\lambda)\frac{t}{mN} v_k H v_k^* . 
\end{align*}	
Over one complete cycle, we therefore obtain
\[
\int_0^{t/m} U_{\lambda}^C(s)^* HU_{\lambda}^C(s) \rmd s
= \frac{t}{m}\big( \lambda H_1 + (1-\lambda) H_0\big),
\]
where
\[
H_0=\frac{1}{N} \sum_{k=1}^N v_k H v_k^*, \qquad
H_1=\frac{1}{N} \sum_{k=1}^N v_{k-1} \int_{0}^1 \gamma_k(s)^* H\gamma_k(s) \rmd s \; v_{k-1}^*
\]
as in Theorem~\ref{th:contDD}, and
\[
F_\lambda\Big(\frac{t}{m}\Big)^m 
= \mathop{\overleftarrow{\exp}} \Big(-\rmi \int_0^{t/m} U_\lambda^C(r)^* H U_\lambda^C(r)\rmd r \Big)^m 
= \mathop{\overleftarrow{\exp}} \Big(-\rmi \int_0^{t} U_\lambda^C(r)^* H U_\lambda^C(r)\rmd r \Big).
\]

We will make use of the bound of the ring~\cite{RWA}:
\[
\| U_a - U_b \|_{\infty} \leq \|S_{ab}\|_{\infty}(1+ \|H_a\|_1 + \|H_b\|_1),
\]
where $U_j(s)$ is the unitary propagator generated by a bounded integrable Hamiltonian $H_j(s)$, $\|\cdot\|_{\infty}$ and $\|\cdot\|_{1}$ are the norms in $L^\infty[0,t]$ and $L^1[0,t]$, respectively, and
\[
S_{ab}(s) = \int_0^s \big(H_a(r)-H_b(r)\big) \rmd r
\]
is the integral action.

In our case, we are interested in
\[
\Big\| F_\lambda\Big(\frac{t}{m}\Big)^m - \rme^{-\rmi t H_\lambda}\Big\|
\]
so let
\[
U_a(t) = F_\lambda\Big(\frac{t}{m}\Big)^m , \qquad U_b(t) = \rme^{-\rmi t ( \lambda H_1 + (1-\lambda) H_0)},
\]
and
\[
H_a(s) = U^C_\lambda(s)^* H U^C_\lambda(s), \qquad H_b(s) = H_\lambda = \lambda H_1+(1-\lambda)H_0.
\]
Thus, if $s$ is somewhere in the $(\ell+1)$-th cycle, i.e., $\ell t/m \leq s \leq (\ell+1)t/m$ where $0\leq\ell\leq m-1$ is an integer, we have
\begin{align*}
\int_0^s H_a(r) \rmd r =& \int_0^s U^C_\lambda(r)^* H U^C_\lambda(r) \rmd r \\
=& \int_0^{\ell t/m} U^C_\lambda(r)^* H U^C_\lambda(r) \rmd r  + \int_{\ell t/m}^s U^C_\lambda(r)^* H U^C_\lambda(r) \rmd r \\
=& \frac{\ell t}{m} \big( \lambda H_1 + (1-\lambda) H_0\big) + \int_{\ell t/m}^s U^C_\lambda(r)^* H U^C_\lambda(r) \rmd r
\end{align*}
and obviously
\[
\int_0^s H_b(r)\rmd s = s \big( \lambda H_1 + (1-\lambda) H_0\big).
\]
Combining the two, we get
\[
S_{ab}(s) = \Big(\frac{\ell t}{m} - s \Big) \big( \lambda H_1 + (1-\lambda) H_0\big)
+ \int_{\ell t/m}^s U^C_\lambda(r)^* H U^C_\lambda(r) \rmd r. 
\]
Since $|s-\ell t/m| \leq t/m$ by construction, we get
\[
\|S_{ab} \|_\infty \leq \frac{t}{m} \|\lambda H_1 + (1-\lambda) H_0 \| + \frac{t}{m}\|H\| \leq \frac{2t}{m}\|H\|.
\]
Moreover,
\[
\|H_a\|_1 = t \| H\|, \qquad  \|H_b\|_1 \leq t \| H\|,
\]
so
\[
\Big\| F_\lambda\Big(\frac{t}{m}\Big)^m - \rme^{-\rmi t \big(\lambda H_1+(1-\lambda) H_0\big)}\Big\| \leq \frac{2t}{m} \| H\| (1+2t \|H\|),
\]
and by the triangle inequality we have
\[
\Big\| F_\lambda\Big(\frac{t}{m}\Big)^m - \rme^{-\rmi t H_0}\Big\| \leq \frac{2t}{m} \| H\| (1+2t \|H\|) + \lambda t \|H_1-H_0\|.
\]
\end{proof}

\begin{remark}
Theorem~\ref{th:ratebounded} tells us how well dynamical decoupling works for given finite values of $m$ and $\lambda$; it works perfectly in the limit $m\to\infty$, $\lambda\to 0$. For a generic decoupling cycle, $H_1\not= H_0$, as shown in Example~\ref{Counterexample}, where $H_1=Y/\pi$, and dynamical decoupling does not work for $\lambda\not= 0$, or, in general, for pulses of width $O(t/(mN))$. However, scaling $\lambda$ together with $m$, e.g.~with $\lambda = O(1/\log m)$, makes it work as can be seen from~\eqref{eq:ratebounded2}. In order to have dynamical decoupling with the optimal convergence rate $O(t/m)$, one should take $\lambda = O(1/m)$, that is a  pulse width of order $O(t/(m^2 N))$. As we shall show in Section~\ref{sec:Euler}, it is possible to achieve convergence even in the case where $\lambda$ remains constant provided we choose the decoupling cycle as a so-called ``Euler cycle", as $H_1=H_0$ in that special case.
\end{remark}

\begin{remark}
The above results hold for an arbitrary cycle $(v_1, v_2, \dots, \unit)$  of unitaries (not necessarily decoupling). While our paper focuses on decoupling, applications to average Hamiltonian theory~\cite{H76} are immediate. 
\end{remark}

\section{Euler dynamical decoupling}\label{sec:Euler}

Suppose $V$ is a projective representation of a group and let $\Gamma\subset V$ be a set of generators of $V$. Recall that the Cayley graph of $(V,\Gamma)$ has got vertices $V$ and at each vertex outgoing edges labelled by $\Gamma$. This means every vertex has also got $|\Gamma|$ incoming edges, and altogether there are $N:= |V| |\Gamma|$ edges in the graph. Let $(v_1,v_2,\ldots,v_N)$ be an Euler cycle of decoupling operations, i.e., a cycle in the Cayley graph of $(V,\Gamma)$ that travels through each edge precisely once. The $k$-th edge is $\gamma_k:=v_{k}^* v_{k-1}\in\Gamma$ and $v_{0} = v_N=\unit$. 
Now for an Euler cycle, each $\gamma_k$ does in fact only depend on the edge and not the corresponding vertex in the Cayley graph; the same is true for its generator $A^{(k)}$, and hence we may also write $A^{(\gamma)}=A^{(k)}$ and $u_\gamma(t;s)=u_k(t;s)$ if $\gamma=\gamma_k$. We use the notation of continuous dynamical decoupling introduced in Section~\ref{sec:prelim}.

We say that, for given decoupling set $V$, Euler cycle, pulse shapes and relative pulse width $\lambda\in(0,1]$, \emph{Euler continuous dynamical decoupling works} if there is a selfadjoint $B$ on $\H_e$ such that
\[
F\Big(\frac{t}{m}\Big)^m\xi \ra \unit\otimes \rme^{- \rmi t B}\xi,\quad m\ra\infty,
\]
for all $\xi\in\H$ and uniformly for all $t$ in compact intervals.
See~\cite{VK03} for the origin of this and applications.

\begin{example}\label{ex:revisited}
To see how Euler cycles can change the situation, let us revisit Example \ref{Counterexample}, for which we saw that continuous dynamical decoupling does not work. There the decoupling set was $V=\{\unit,X,Y,Z\}$ and $H=X$ and $\lambda=1$; a natural choice of an Euler cycle would be $(Z,Y,X,\unit,X,Y,Z,\unit)$. Going through the discussion of Example \ref{Counterexample}, we obtain the overall time evolution
\begin{align*}
F_1\Big( \frac{t}{m}\Big) ^m =& \Big( 
\rme^{-\rmi \frac{\pi}{2}Z -\rmi \frac{t}{8m}H}  
\rme^{-\rmi \frac{\pi}{2}X -\rmi \frac{t}{8m}H}
\rme^{-\rmi \frac{\pi}{2}Z -\rmi \frac{t}{8m}H}  
\rme^{-\rmi \frac{\pi}{2}X -\rmi \frac{t}{8m}H}\\
& \quad \cdot\rme^{-\rmi \frac{\pi}{2}X -\rmi \frac{t}{8m}H}
\rme^{-\rmi \frac{\pi}{2}Z -\rmi \frac{t}{8m}H}  
\rme^{-\rmi \frac{\pi}{2}X -\rmi \frac{t}{8m}H}
\rme^{-\rmi \frac{\pi}{2}Z -\rmi \frac{t}{8m}H}  
\Big)^{m}\\
=& \Big(
\big((-\rmi Y - \frac{t}{2\pi m}\unit + \rmi \frac{t}{4m}Z + O(\frac{1}{m^2})\big)^2
\big((\rmi Y - \frac{t}{2\pi m}\unit + \rmi \frac{t}{4m}Z + O(\frac{1}{m^2})\big)^2
\Big)^m\\
=& \Big( \unit + O(\frac{1}{m^2})
\Big)^m \to \unit, \qquad m\to\infty ,
\end{align*}
so continuous dynamical decoupling with the above Euler cycle works in this case.
\end{example}

We generalise this in the following

\begin{theorem}\label{th:Euler}
Under the assumptions of Theorem~\ref{th:contDD} and if the given decoupling cycle is an Euler cycle, Euler continuous dynamical decoupling works.
\end{theorem}

\begin{proof}
It follows from the property of the Euler cycle that
\begin{align*}
\sum_{k=1}^N v_k  & \Big(\frac{\lambda}{N} \int_{0}^{1} \gamma_{k+1}(s)^* H \gamma_{k+1}(s) \rmd s + \frac{1-\lambda}{N} H\Big)v_k^*\\
=& \sum_{v\in V}\sum_{\gamma\in\Gamma} v \Big( \frac{\lambda}{N} \int_{0}^{1} \gamma(s)^* H \gamma(s) \rmd s + \frac{1-\lambda}{N} H\Big)v^*\\
=& \frac{1}{N} \sum_{v\in V} v \sum_{\gamma\in\Gamma} \Big(\lambda \int_{0}^{1} \gamma(s)^* H \gamma(s) \rmd s + \frac{1-\lambda}{N} H\Big)v^*,
\end{align*}
which we know from Theorem~\ref{th:contDD} to be essentially selfadjoint on $\D$. Here the first equality follows from the particular structure of the Euler cycle. It then follows from~\cite[Thm.3.1]{ABFH} and the decoupling property of $V$ in Definition~\ref{def:dd} that the closure of the right-hand side is of the form $H_\lambda=\unit\otimes B$, and that
 \[
F_\lambda\Big(\frac{t}{m}\Big)^m\xi \ra \rme^{-\rmi t H_\lambda}\xi,\quad m\ra\infty,
\]
uniformly for $t$ in compact intervals and all for $\xi\in\H$.
\end{proof}

\begin{remark}\label{rem:Eulerratebounded}
We see that Euler dynamical decoupling works independently of the choice of $\lambda$.

Moreover, in the case of bounded $H$, one can show that $H_1=H_0$ and hence $H_\lambda=H_0$, independently of $\lambda$. Thus Theorem~\ref{th:ratebounded} yields
\[
\Big\| F_\lambda\Big(\frac{t}{m}\Big)^m - \rme^{-\rmi t H_0}\Big\| \leq \frac{2t}{m} \| H\| (1+2t \|H\|).
\]

In order to prove $H_1=H_0$, we notice that by construction \[
H_0=\frac{1}{|V|}\sum_{v\in V} vHv^* =\unit_{\H_s}\otimes \frac{1}{\dim\H_s}\tr_{\H_s} H,
\]
where $\tr_{\H_s} H$ denotes the partial trace of $H$ with respect to $\H_s$. We then calculate as follows:
\begin{align*}
H_1 =& \frac{1}{N}\sum_{k=1}^N  v_{k-1} \int_{0}^1 \gamma_k(s)^*H\gamma_k(s) \rmd s \; v_{k-1}^* = \frac{1}{|V||\Gamma|}\sum_{v\in V} \sum_{\gamma\in\Gamma} v \int_{0}^1 \gamma(s)^*H\gamma(s) \rmd s \; v^* 
\\
=& \frac{1}{|\Gamma|} \sum_{\gamma\in\Gamma} \frac{1}{|V|}\sum_{v\in V}  v \int_{0}^1 \gamma(s)^*H\gamma(s) \rmd s \; v^* 
= \frac{1}{|\Gamma|} \sum_{\gamma\in\Gamma} \unit_{\H_s} \otimes \frac{1}{\dim\H_s}\tr_{\H_s} \Big(\int_{0}^1 \gamma(s)^*H\gamma(s) \rmd s\Big)\\
=&\frac{1}{|\Gamma|} \sum_{\gamma\in\Gamma} \unit_{\H_s} \otimes \frac{1}{\dim\H_s}\tr_{\H_s} H
= H_0,
\end{align*}
where the last equality follows from the fact that $\gamma(s)$ is a unitary acting only on $\H_s$ and hence the partial trace $\tr_{\H_s}$ is invariant under conjugation by $\gamma(s)$.

If $H$ is unbounded but $\D(H)$ factorises, namely it is of the form $\H_s\otimes\D_e$, we can still use the Schmidt decomposition to get a finite sum
\[
H=\sum_{k} H_{s,k}\otimes H_{e,k},
\]
which is essentially selfadjoint on $\H_s\otimes\D_e$. In that case we can use the partial traces again and the same argument goes through.
\end{remark}

\begin{example}\label{th:finiteEuler}
If $H$ factorizes as $H=H_s\otimes H_e$ and is selfadjoint on $\D(H)=\H_s\otimes\D(H_e)$ then, for any choice of decoupling set, generator set, Euler cycle and pulse shape, Euler continuous dynamical decoupling works: It can easily be seen that the conditions in Theorem~\ref{th:contDD} are fulfilled with $\D=\D(H)$ and $H_\lambda= \tr(H_s)\unit_{\H_s}\otimes H_e$.
\end{example}

Example~\ref{th:finiteEuler} covers a variety of models but many other models do not fall into this class, yet Theorem~\ref{th:Euler} can be applied, such as in the following:

\begin{example}\label{ex:deeppocket}
We consider the ``deep-pocket model" developed in~\cite{FGMS}, namely we let 
\[
\H'=\C^2\otimes L^2(\R_-)
\]
and $H'=Z\otimes \rmi\frac{\rmd}{\rmd x}$ on the domain
\[
\D(H') = \{ \phi\in\C^2\otimes H^1(\R_-): \; \phi(0)=(X\otimes\unit)\phi(0)\}.
\]
Physically, this describes a spin-1/2 particle moving along the negative half-axis and being reflected at $0$ at the price of a spin flip. A suitable (canonical and reduced, see Remark \ref{rem:dd-def}) decoupling set consists of $V'=\{X\otimes\unit, \unit\otimes\unit\}$.

It is evident that the domain $\D(H')$ of the model does not factorize as a tensor product, so Example~\ref{th:finiteEuler} does not apply here. However, $\D(H')$ is invariant under $V'$ and one realizes that $H'+ XH'X= 0$.  Let us define
\[
F(t)= \rme^{-\rmi \frac{t}{2} H} \rme^{-\rmi \frac{t}{2} XHX}, \quad t\in\R.
\]
Then
\[
\frac{F(t)-\unit}{t}\ra \frac{\rmi}{2} (H + X H X) =0,
\]
which is essentially selfadjoint on $\D(H)$, so~\cite[Thm.3.1]{ABFH} implies bangbang dynamical decoupling works.

In order to treat Euler continuous dynamical decoupling, it is useful to pass to a unitarily equivalent version of this model, cf.~\cite{FGMS}: namely, let
\[
\H= L^2(\R)
\]
and $H= \rmi\frac{\rmd}{\rmd x}$ on the domain $\D(H)=H^1(\R)$, which is selfadjoint. The corresponding decoupling set is $V=\{X_0, \unit\}$, where
\[
(X_0 \xi)(x) = \xi(-x), \quad x\in\R.
\]
Let $\D= \D(H)$. We see that $\D$ is invariant under $V$ and that
\begin{equation}\label{eq:X0Hn}
H^n X_0 = (-1)^n X_0 H^n.
\end{equation}

We use $\Gamma=\{\gamma=X_0\}$ and $V=\{\unit,X_0\}$, 
and the continuous dynamical decoupling operation is modelled by the function 
\[
\gamma(s) = \rme^{-\rmi \varphi(s) \frac{\pi}{2} (X_0-\unit)}
= \rme^{\rmi\varphi(s)}\cos\big(\varphi(s)\frac{\pi}{2}\big) \unit 
- \rmi \rme^{\rmi\varphi(s)}\sin\big(\varphi(s)\frac{\pi}{2}\big) X_0 ,
\qquad s\in [0,1],
\]
where $\varphi$ is a fixed continuous function such that $\varphi(0)=0$ and $\varphi(1)=1$. We immediately see that $\gamma(1)=X_0$, that $\gamma(s)$ and leaves $\D$ invariant, and that $\|H\gamma(s)\xi\|\leq 2\|H\xi\|$ as follows from~\eqref{eq:X0Hn}, and hence assumption (i) of Theorem~\ref{th:contDD} is fulfilled. 

In order to verify (ii), note that
\begin{align*}
(\gamma(s)^* H \gamma(s)) \xi(x) =& 
\rmi \cos^2\big( \varphi(s) \frac{\pi}{2}\big) \xi'(x) - 2 \sin\big( \varphi(s) \frac{\pi}{2}\big)\cos\big( \varphi(s) \frac{\pi}{2}\big) \xi'(-x)  - \rmi\sin^2\big( \varphi(s) \frac{\pi}{2}\big) \xi'(x)\\
=& \rmi \cos( \varphi(s) \pi) \xi'(x) - \sin( \varphi(s) \pi)\xi'(-x)\\
=& \cos( \varphi(s) \pi) (H\xi) (x) +\rmi \sin( \varphi(s) \pi)(X_0H \xi)(x)
\end{align*}
and hence
\begin{align*}
\int_{0}^1 \gamma(s)^* H \gamma(s) \xi \rmd s 
=& c_\unit H\xi + c_{X_0} X_0H\xi,
\end{align*}
with certain $c_\unit,c_{X_0}\in\C$. Finally, for all $\xi\in\D$, we have
\begin{align*}
\sum_{v\in V}\sum_{\gamma\in\Gamma} v \Big( \gamma \int_{0}^1 \gamma(s)^* H \gamma(s)\rmd s \Big) v^* \xi 
=& \sum_{v\in V} v \Big(c_\unit X_0 H + c_{X_0} H \Big)  v^*\xi \\
=& \Big(c_\unit X_0 H + c_{X_0} H + c_\unit H X_0 + c_{X_0} X_0 H X_0 \Big)  v^*\xi = 0,
\end{align*}
where the last equality follows from~\eqref{eq:X0Hn} again, so
\[
\sum_{v\in V}\sum_{\gamma\in\Gamma} v \Big( \lambda \gamma \int_{0}^1 \gamma(s)^* H \gamma(s)\rmd s + (1-\lambda) H \Big) v^* =0
\]
is essentially selfadjoint on $\D$, proving (ii).
\end{example}

\section*{Acknowledgments}

DB acknowledges funding by the Australian Research Council (project
numbers FT190100106, DP210101367, CE170100009).
PF is partially supported by the Italian National Group of Mathematical Physics (GNFM-INdAM), by Istituto Nazionale di Fisica Nucleare (INFN) through the project ``QUANTUM'', and by Regione Puglia and QuantERA ERA-NET Cofund in Quantum Technologies (GA No. 731473), project PACE-IN.

\section*{Final statements}

On behalf of all authors, the corresponding author states that there is no conflict of interest. 

Data sharing is not applicable to this article as no new data were created or analyzed in this study.

\end{document}